\documentclass[conference]{IEEEtran}

\usepackage[letterpaper, left=0.680in, right=0.680in, bottom=0.990in, top=0.71in]{geometry}
\usepackage{ifpdf}
\ifCLASSINFOpdf
\usepackage[pdftex]{graphicx}
\graphicspath{{./figures/}}
\DeclareGraphicsExtensions{.pdf,.jpeg,.png}
\else
\usepackage[dvips]{graphicx}
\DeclareGraphicsExtensions{.pdf}
\usepackage[caption=false, font=footnotesize]{subfig}
\fi

\usepackage{amsthm}
\usepackage{balance}
\usepackage{epstopdf}
\usepackage{amsmath}
\usepackage{amssymb}
\usepackage{color}
\usepackage{array}
\usepackage{mathtools}
\usepackage{enumitem}
\usepackage[mathcal]{euscript}
\usepackage{tabularx}
\usepackage{multirow}
\usepackage{booktabs}
\usepackage{makecell}
\usepackage{aligned-overset}
\usepackage{subfig}
\usepackage{url}
\usepackage{tikz}
\usetikzlibrary{arrows.meta,calc,positioning,angles,quotes,backgrounds,shapes.geometric}

\definecolor{tndblue}{HTML}{0018A8}\definecolor{tndpurple}{HTML}{6B3FA0}
\definecolor{tndgreen}{HTML}{086C46}\definecolor{tndgray}{HTML}{6E6E6E}
\definecolor{tndred}{HTML}{ED1B24}\definecolor{tndcyan}{HTML}{00A9CE}
\definecolor{tndgold}{HTML}{D4AA00}\definecolor{tndorange}{HTML}{FF9800}

\newtheorem{lemma}{Lemma}

\newtheorem{definition}{Definition}

\newtheorem{proposition}{Proposition}

\newtheorem{remark}{Remark}

\begin{document}

\bstctlcite{IEEEtran:BSTcontrol}

\title{CSI-Aging-Aware Functional Split Optimization \\ and KAN-LSTM Prediction for Open-RAN NTN}

\author{
  \IEEEauthorblockA{Alexandar~Tsvetanov~Dimitrov, Yosefine Triwidyastuti, Eray Guven, \\ Tri~Nhu~Do, and Gunes~Karabulut~Kurt}
  \IEEEauthorblockA{Department of Electrical Engineering, Polytechnique Montréal, Montreal, QC, Canada}
  \IEEEauthorblockA{\{alexandar.dimitrov, yosefine-2.triwidyastuti, guven.eray, tri-nhu.do, gunes.kurt\}@polymtl.ca}
}

\maketitle

\begin{abstract}
In this paper, we study functional-split selection for an Open radio access network (RAN) deployed as a non-terrestrial network (NTN), in which a low-Earth-orbit (LEO) satellite hosts the radio unit, so that each candidate split (7.2x, 6, 2) places the RAN at a distinct point of the transparent-to-regenerative axis. We capture the space--ground geometry exactly through an instantaneous link plane, which yields closed-form slant range, elevation, delay, and Doppler for an inclined orbit and a rotating Earth. On this geometry we derive the received-SNR density of the elevation-driven LOS/NLOS mixture channel and, through a split-dependent channel state information (CSI) aging mechanism, an uncoded bit-error-rate (BER) law for each split, establishing that what separates the splits is the freshness of the CSI each brings to its adaptation loops: a loop closed on the ground pays the feeder round-trip and incurs an aging-induced error floor. We formulate split selection as a prefer-centralized feasibility problem over reliability, aging, and fronthaul-capacity constraints, and we show that its exact solution is a closed-form region rule in four state coordinates, validated against exhaustive enumeration without a single mismatch. Because this solution is reactive whereas a split change must be initiated ahead of time, we propose a compact Kolmogorov--Arnold network (KAN)-LSTM, sized for an edge xApp/rApp, that predicts the split schedule seconds ahead; we show that it attains high schedule accuracy at a sub-percent BER-violation probability and delivers nearly all of the ground-truth policy's goodput, outperforming naive and reactive baselines.
\end{abstract}


\section{Introduction}
\label{sec:intro}
Non-terrestrial networks (NTN) built on low-Earth-orbit (LEO) satellites are a central ingredient of 6G coverage extension \cite{Jamshed_MCOM_2026}, and the Open radio access network (RAN) architecture, i.e., O-RAN, with its disaggregated O-RU/O-DU/O-CU functions, open fronthaul, and RAN Intelligent Controllers (RIC), is the natural framework in which to integrate them with the terrestrial segment \cite{Polese_COMST_2023}. Each standardized split (7.2x, 6, or 2) \cite{Alam_COMST_2026}, whose boundaries the emerging 3GPP 6G RAN architecture is refining further \cite{Lin2026_arXiv}, places the configuration at a distinct point of the \emph{transparent-to-regenerative} (bent-pipe-to-onboard) axis: split 7.2x keeps almost all baseband processing on the ground, split 6 moves the physical layer onboard, and split 2 relocates the entire distributed unit onto the payload.

In terrestrial deployments the split is largely a static dimensioning choice \cite{Alam_COMST_2026}; over a LEO pass it cannot be. Slant range, propagation delay, Doppler, and fading vary continuously, and, decisively, so does the freshness of the channel state information (CSI) each split brings to its physical- and MAC-layer adaptation loops. A transparent, ground-processed split closes those loops through the feeder, so its CSI is milliseconds old and, under fast Doppler, induces a bit-error-rate (BER) floor that no transmit power can remove; a fully regenerative split avoids this but forfeits the pooling, upgradeability, and low payload burden that motivate centralization.

Adaptive split selection has been pursued along two lines: rule-based optimization and learning-based methods. The former casts the split as a constrained orchestration problem, minimizing deployment cost and power under compute, link-capacity, and QoS constraints \cite{Laouadi_ICC_2025}; the latter, chiefly deep reinforcement learning, optimizes split or DU placement under limited compute and energy resources \cite{Murti_TWC_2022}. In NTN, whose terrestrial integration is still consolidating through standardization and experimental testbeds \cite{Carbonara_OJCOMS_2025}, the few dynamic treatments place security functions across the constellation under compute and visibility constraints \cite{Mahyoub_IWCMC_2025}, while AI-native Open RAN blueprints forecast link key performance indicators (KPIs) but leave the split itself fixed \cite{Do_RIVF_2025}. Two gaps remain: none couples the decision to the physical mechanism that actually separates the splits over a satellite link, namely CSI aging, so reliability enters only as an exogenous penalty rather than a derived consequence of the placement; and all are \emph{reactive}, whereas a split change is a reconfiguration whose control signaling itself traverses the feeder and must be initiated ahead of time, so the decision must be \emph{predicted}, not merely computed.

In this paper, taking the access-link channel-state vector as the decision variable, we derive how each split's CSI age sets its reliability, cast split selection as a prefer-centralized feasibility problem over reliability, aging, and fronthaul-capacity constraints, and solve it \emph{exactly}: the optimal policy is a closed-form region rule in four state coordinates, not an intractable program. As this rule is reactive, we then train a compact KAN-LSTM \cite{liu2025kan} to forecast its schedule ahead of time from the state trajectory, giving an anticipatory policy deployable as an edge xApp/rApp. Our main contributions are summarized as follows:
\begin{itemize}[leftmargin=*]
\item We propose a function-aware model that separates physical nodes from O-RAN functions through a placement mapping $\varphi_s$ over splits 7.2x/6/2 along the transparent--regenerative axis, together with a three-dimensional link-plane geometry yielding closed-form slant range, elevation, delay, and Doppler, which induces the access-link state vector $\vec{z}_t$.
\item We derive the received-SNR density of the elevation-driven LOS/NLOS mixture channel and, through a split-dependent CSI-aging mechanism, an uncoded-BER law for each split. We prove that stale CSI induces a BER floor with closed-form per-split SNR thresholds, so that reliability becomes a derived consequence of the placement rather than an exogenous penalty.
\item We formulate split selection as a prefer-centralized feasibility problem over reliability, aging, and fronthaul-capacity constraints, and we show that its exact solution is a closed-form region rule in $(\Gamma_t,\tau_{p,t},f_{d,t},B_t)$.
\item We propose a compact, xApp/rApp-deployable KAN-LSTM that forecasts the ground-truth schedule seconds ahead, so that the reconfiguration can be initiated within its reconfiguration lead time.
\item Results show that both the per-split BER thresholds and the region rule agree with direct evaluation and exhaustive enumeration without a single mismatch over a large Monte Carlo ensemble, and that the predictor attains high schedule accuracy at a sub-percent BER-violation probability while delivering nearly all of the ground-truth policy's goodput, far above random and fixed always-centralized baselines.
\end{itemize}
\section{System Model}
\label{sec:system}

\begin{figure}[t]
	\centering
	\resizebox{\columnwidth}{!}{%
		\begin{tikzpicture}[
			>=Latex, line cap=round, line join=round, font=\small,
			every node/.style={align=center},
			oran/.style={draw, rounded corners=2pt, fill=tndblue!8, font=\footnotesize, inner sep=3pt, minimum height=0.6cm},
			feeder/.style={->, semithick, densely dashed, tndgray},
			splitlink/.style={->, line width=1.2pt, densely dashed, tndred},
			f1/.style={->, semithick, densely dashed, tndgray}
			]
			\colorlet{earthblue}{tndcyan!18}\colorlet{earthedge}{tndblue}
			\colorlet{coneorange}{tndorange}\colorlet{footgreen}{tndgreen}
			\colorlet{planepurple}{tndpurple!45}\colorlet{satred}{tndred}
			\colorlet{ueblue}{tndcyan}\colorlet{orbitred}{tndred}
			\coordinate (O) at (0,0);\coordinate (S) at (0,5.6);
			\coordinate (FL) at (-1.85,2.35);\coordinate (FR) at (1.85,2.35);\coordinate (FC) at (0,2.35);
			\coordinate (U1) at (-1.10,2.18);\coordinate (U2) at (-0.45,2.47);
			\coordinate (U) at (0.55,2.18);\coordinate (U4) at (1.15,2.37);\coordinate (U5) at (0.08,1.98);
			\coordinate (Gs) at (-2.30,1.93); 
			\coordinate (Dn) at (-5.75,1.15);\coordinate (Cn) at (-4.20,0.35);
			\coordinate (P1) at (-2.65,-0.60);\coordinate (P2) at (2.90,-0.20);
			\coordinate (P3) at (2.35,6.75);\coordinate (P4) at (-3.15,7.15);
			\begin{scope}[on background layer]
				\fill[planepurple,opacity=0.20] (P1)--(P2)--(P3)--(P4)--cycle;
				\draw[planepurple!75!black,dashed,thin, dotted] (P1)--(P2)--(P3)--(P4)--cycle;
				\draw[orbitred,thin, dotted] (-5.50,0) arc[start angle=180,end angle=0,x radius=5.50,y radius=5.60];
			\end{scope}
			\node[planepurple!75!black,font=\bfseries,rotate=0] at (4,2.50) {instantaneous link plane};
			\node[planepurple!75!black,font=\bfseries] at (-1.60,6.60) {$\Pi(t)=\mathrm{span}\{\vec{r}_U(t),\vec{r}_S(t)\}$};
			\fill[earthblue] (-3,0) arc[start angle=180,end angle=0,radius=3]--cycle;
			\draw[earthedge,thin, dotted] (-3,0) arc[start angle=180,end angle=0,radius=3];
			\draw[earthedge!70,thick] (-3,0)--(3,0);
			\fill[black] (O) circle[radius=2.2pt];
			\node[below=2pt of O,font=\bfseries] {$O$};\node[below=8pt of O] {Earth center};
			\fill[coneorange,opacity=0.30] (S)--(FL) arc[start angle=180,end angle=0,x radius=1.85,y radius=0.48]--cycle;
			\draw[coneorange!90!black,very thick] (S)--(FL);\draw[coneorange!90!black,very thick] (S)--(FR);
			\draw[coneorange!90!black,dash dot,thick,->] (S)--(FC);
			\pic[draw=coneorange!90!black,very thick,->,angle radius=0.85cm,angle eccentricity=1.30,"$\mathrm{FOV}=2\theta_b$"]{angle=FL--S--FR};
			\fill[footgreen,opacity=0.38] (FC) ellipse[x radius=1.85,y radius=0.48];
			\draw[footgreen,very thick] (FC) ellipse[x radius=1.85,y radius=0.48];
			\node[footgreen!60!black,font=\bfseries] at (0,3.08) {coverage footprint};
			\fill[satred] (S) circle[radius=3pt];
			\draw[satred,thick,fill=white] ($(S)+(-0.20,-0.10)$) rectangle ($(S)+(0.20,0.10)$);
			\draw[satred,thick,fill=tndblue!15] ($(S)+(-0.72,-0.08)$) rectangle ($(S)+(-0.22,0.08)$);
			\draw[satred,thick,fill=tndblue!15] ($(S)+(0.22,-0.08)$) rectangle ($(S)+(0.72,0.08)$);
			\node[above=7pt of S,font=\bfseries] {LEO satellite $S$ (O-RU)};
			\foreach \P in {U1,U2,U4,U5}{\fill[ueblue] (\P) circle[radius=2.0pt];\draw[black,thin] (\P) circle[radius=2.0pt];}
			\node[ueblue!70!black,font=\footnotesize] at (-1.55,1.55) {other UEs in footprint};
			\fill[tndgold] (U) circle[radius=3pt];\draw[black,thick] (U) circle[radius=3pt];
			\node[above right=0pt,font=\bfseries] at (U) {user $U$};
			\draw[black!55,thick] (O)--(S);\draw[black!70,thick] (O)--(U);
			\fill[planepurple!55,opacity=0.10] (O)--(S)--(U)--cycle;
			\draw[planepurple!80!black,thick] (O)--(S)--(U)--cycle;
			\draw[black,very thick] (S)--(U);
			\node[right,font=\bfseries] at ($(S)!0.18!(U)+(0.44,0)$) {$(S,U)$ access, $d_{\mathrm{SU}}(t)$};
			\node[left] at ($(O)!0.66!(S)+(-0.10,0)$) {$OS$};
			\node[right] at ($(O)!0.55!(U)+(0.10,-0.05)$) {$OU$};
			\node[black!70,font=\footnotesize,above right=0pt] at ($(U)+(0.06,0.22)$) {$\alpha_t$};
			\begin{scope}[on background layer]
				\draw[dotted,thick,rounded corners=3pt] (-6.70,-0.55) rectangle (-3.35,2.20);
			\end{scope}
			\node[font=\bfseries] at (-5.00,1.92) {Ground gateway $G$};
			\node[oran] (D) at (Dn) {O-DU\\$D$};\node[oran] (C) at (Cn) {O-CU\\$C$};
			\draw[f1] (C) -- node[font=\footnotesize,above,pos=0.5,xshift=9pt] {F1:\ $C\!\to\!D$} (D);
			\fill[black] (Gs) circle[radius=2.2pt];\draw[black,thick] (Gs)--++(0,0.30);
			\node[above=0pt,font=\footnotesize] at ($(Gs)+(0,0.30)$) {gateway RF};
			\draw[tndgray,thick] (C.east) .. controls (-3.5,1.4) .. (Gs);
			\node[tndgray,font=\footnotesize] at (-3.05,0.75) {co-located};
			\draw[splitlink] (D.north) to[bend left=10]
			node[font=\footnotesize,left=1pt,pos=0.66,text=satred] {feeder $(G,S)$: split-$s_t^\star$ interface} ($(S)+(-0.60,-0.35)$);
		\end{tikzpicture}%
	}
	\caption{\small Function-aware O-RAN NTN topology: the satellite $S$ hosts the O-RU and the gateway $G$ the O-CU and, depending on the split, the O-DU, the feeder $(G,S)$ transporting the split-$s_t^\star$ interface, while the instantaneous link plane $\Pi(t)$ contains the $O$--$S$--$U$ triangle fixing $d_{\mathrm{SU}}(t)$, $\alpha_t$, $\tau_{p,t}$, and $f_{d,t}$.}
	\label{fig:topology}
\end{figure}
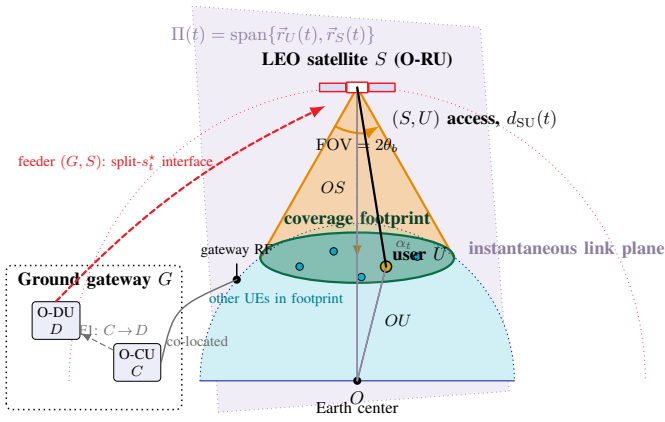
\subsection{Function-Aware Architecture and Placements}
We consider an O-RAN--NTN, as depicted in Fig.~\ref{fig:topology}. The physical node set is $\mathcal{V}_{\mathrm{phy}}=\{S,G,U\}$ (satellite, gateway, user terminal) and the O-RAN functional set is $\mathcal{V}_{\mathrm{fun}}=\{R,D,C\}$ (O-RU, O-DU, O-CU). A split-dependent placement mapping $\varphi_s:\mathcal{V}_{\mathrm{fun}}\rightarrow\mathcal{V}_{\mathrm{phy}}$, $s\in\mathcal{S}=\{s_1,s_2,s_3\}=\{7.2\mathrm{x},6,2\}$, assigns functions to hosts, with $\varphi_s(R)=S$ and $\varphi_s(C)=G$ for every candidate; moving from 7.2x toward 2 relocates progressively more PHY/MAC processing onto the payload, until $\varphi_{s_3}(D)=S$ places the entire distributed unit onboard. A split is therefore a \emph{placement decision}, not a choice among parallel interfaces on a fixed architecture. The system contains exactly two physical links: the feeder $(G,S)$, transporting the split-induced logical interface (Open Fronthaul, MAC--PHY, or F1), and the NTN access downlink $(S,U)$, the only link simulated as a wireless channel. Annotating each function with its host ($X_{\mathsf{G}}$ or $X_{\mathsf{S}}$), the three downlink chains are
\begin{align}
s_1:&\ \ C_{\mathsf{G}} \xrightarrow{\text{F1 (int.)}} D_{\mathsf{G}} \xrightarrow[\text{feeder}]{\text{Open FH}} R_{\mathsf{S}} \xrightarrow{\text{access}} U, \label{eq:chain-s1}\\
s_2:&\ \ C_{\mathsf{G}} \longrightarrow M_{\mathsf{G}} \xrightarrow[\text{feeder}]{\text{MAC--PHY}} P_{\mathsf{S}} \xrightarrow{\text{access}} U, \label{eq:chain-s2}\\
s_3:&\ \ C_{\mathsf{G}} \xrightarrow[\text{feeder}]{\text{F1}} D_{\mathsf{S}} \longrightarrow R_{\mathsf{S}} \xrightarrow{\text{access}} U, \label{eq:chain-s3}
\end{align}
where under Split 6 the O-DU is partitioned into a ground-hosted RLC/MAC block $M$ (scheduling, HARQ control, link adaptation) and an onboard complete-PHY block $P$ (estimation, coding, (de)modulation, precoding, FFT, RF). Split 7.3 retains the ground-hosted PHY of 7.2x, hence is not separable from it under the CSI-aging mechanism below.

\subsection{Three-Dimensional Link-Plane Geometry}
The Earth-centered frame has origin $O$, Earth radius $R_\oplus=6371$~km, altitude $h_0=600$~km, and carrier $f_c=2$~GHz (3GPP FR1-NTN, handheld access), with $\|\vec{r}_S(t)\|=R_\oplus+h_0$, $\|\vec{r}_U(t)\|=R_\oplus$. The orbit may be arbitrarily inclined and the user arbitrarily located, the gateway is fixed at $\vec{r}_G$, and the circular orbit is traced at $V=7.5622$~km/s. At every instant, $O$, $U$, and $S$ span the unique \emph{instantaneous link plane} $\Pi(t)=\mathrm{span}\{\vec{r}_U(t),\vec{r}_S(t)\}$, in which the entire line-of-sight (LOS) geometry of $(S,U)$ exists. With the geocentric angle $\psi(t)$ given by $\cos\psi(t)=\vec{r}_U(t)^\top\vec{r}_S(t)/[R_\oplus(R_\oplus+h_0)]$, the law of cosines in $\Pi(t)$ yields the slant range in closed form,
\begin{equation}
\label{eq:slant-range}
d_{\mathrm{SU}}(t)=\sqrt{R_\oplus^2+(R_\oplus+h_0)^2-2R_\oplus(R_\oplus+h_0)\cos\psi(t)},
\end{equation}
exact in three dimensions. With $\vec{d}_{\mathrm{SU}}(t)=\vec{r}_S(t)-\vec{r}_U(t)$ and the local zenith $\vec{n}_U(t)=\vec{r}_U(t)/\|\vec{r}_U(t)\|$, the elevation is $\alpha(t)=\arcsin\bigl(\vec{d}_{\mathrm{SU}}(t)^\top\vec{n}_U(t)/\|\vec{d}_{\mathrm{SU}}(t)\|\bigr)$, in one-to-one correspondence with $d_{\mathrm{SU}}(t)$ for fixed altitude over the $10^\circ$-masked range. The satellite velocity $\vec{v}_S(t)$ lies in the orbital plane with $\|\vec{v}_S(t)\|=V$, and Earth rotation induces the user velocity $\vec{v}_U(t)=\vec{\omega}_E\times\vec{r}_U(t)$.

\subsection{Signal Model and Observed State}
Throughout, a subscript $t$ denotes the value of a quantity at epoch $t$, whereas the argument $(t)$ is reserved for the continuous-time geometry of Sec.~\ref{sec:system}-B. The downlink waveform is CP-OFDM per 3GPP NR-NTN with per-subcarrier BPSK for analytical tractability. Symbols and fading are unit-power normalized, and antenna/beam effects are folded into fixed gains $G_{\mathrm{tx}}$, $G_{\mathrm{rx}}$. With sufficient cyclic prefix and NR-NTN timing/frequency precompensation, one resource element obeys
\begin{equation}
\label{eq:scalar-re}
y_t = \sqrt{P_t\,G_{\mathrm{tx}}\,G_{\mathrm{rx}}}\;10^{-PL_t/20}\,h_t\,x_t + n_t,
\end{equation}
where $P_t$ is the transmit power, $PL_t$ the total path loss in dB, $h_t$ the frequency-flat fading coefficient, and $n_t\sim\mathcal{CN}(0,N_0 B_t)$, with $N_0$ the noise power spectral density and $B_t$ the occupied bandwidth. The SNR of the received resource element is $\Gamma_t=\mathrm{SNR}_t=P_t G_{\mathrm{tx}} G_{\mathrm{rx}} |h_t|^2 10^{-PL_t/10}/(N_0 B_t)=\beta_t |h_t|^2$, with $\beta_t$ collecting the deterministic large-scale terms.

\begin{definition}[Observed access-link state]
\label{def:state-vector}
At time $t$, the observed access-link state of $(S,U)$ is $\vec{z}_t = [\,|h_t|, \mathrm{SNR}_t, B_t, \alpha_t, \tau_{p,t}, f_{d,t}\,]^\top \in \mathbb{R}^6$, where $\tau_{p,t}$ is the one-way propagation delay and $f_{d,t}$ the Doppler shift.
\end{definition}
The feeder $(G,S)$ is not modeled as a stochastic channel; it enters only through the transport rate $R_{\mathrm{fh}}(s_t)$ and the CSI-age budgets below.

\section{Channel Statistics and CSI-Aging Analysis}
\label{sec:analysis}
\subsection{LOS/NLOS Mixture Channel}
Let $L_t\in\{0,1\}$ denote the propagation state (1 for LOS and 0 for NLOS). Conditioned on elevation and environment class $e_t$, $L_t \sim \mathrm{Bernoulli}(p_{\mathrm{LOS}}(\alpha_t,e_t))$. The path loss decomposes as $PL_t = FSPL(d_{\mathrm{SU}}(t),f_c) + SF_t + CL_t + PL_g + PL_s$, where $FSPL(\cdot)$ is the free-space path loss at the exact slant range \eqref{eq:slant-range}, $SF_t$ the lognormal shadowing with state-dependent variance, $CL_t$ the clutter loss, active only under NLOS, and $PL_g$ and $PL_s$ the coarse gaseous-absorption and scintillation terms, respectively. The unit-power small-scale amplitude is state-conditioned,
$
|h_t| \mid (L_t=1,\alpha_t) \sim \mathrm{Rician}\!\bigl(K(\alpha_t)\bigr),
|h_t| \mid (L_t=0) \sim \mathrm{Rayleigh},
$
so that, conditioned on $(\alpha_t,e_t)$, the fading law is the elevation-driven mixture of the Rician density $f_{|h|\mid L=1}(h)=2(1+K_t)h\,e^{-K_t-(1+K_t)h^2} I_0(2h\sqrt{K_t(1+K_t)})$ and the Rayleigh density $f_{|h|\mid L=0}(h)=2h\,e^{-h^2}$, weighted by $p_{\mathrm{LOS}}(\alpha_t,e_t)$.

\begin{lemma}
\label{lem:snr}
Conditioned on $(\alpha_t,e_t,L_t,SF_t,B_t)$, the pdf of $\Gamma_t=\beta_t|h_t|^2$ is
$f_{\Gamma_t}(\gamma)=\tfrac{1}{2\sqrt{\beta_t\gamma}} f_{|h_t|}(\sqrt{\gamma/\beta_t})$, $\gamma\ge0$. In particular,
$f_{\Gamma_t \mid L_t=1}(\gamma)
= \tfrac{1+K_t}{\beta_t}
e^{-K_t-(1+K_t)\gamma/\beta_t}
I_0\bigl(2\sqrt{K_t(1+K_t)\gamma/\beta_t}\bigr)$
and
$f_{\Gamma_t \mid L_t=0}(\gamma)
= \tfrac{1}{\beta_t}e^{-\gamma/\beta_t}$,
and, conditioned only on $(\alpha_t,e_t,SF_t,B_t)$, the SNR pdf is the $p_{\mathrm{LOS}}$-weighted mixture of the two conditional laws, each branch evaluated with its own $\beta_t$ (state-dependent $CL_t$ and shadowing variance).
\end{lemma}
\begin{proof}
We use the monotone map $\gamma=\beta_t h^2$, for which
$
\frac{dh}{d\gamma}=\frac{1}{2\sqrt{\beta_t\gamma}}.
$
We apply the change-of-variables formula and average over the Bernoulli state to obtain the result. Due to space limitations, we omit the intermediate derivation steps.
\end{proof}

\begin{figure}[t]
\centering
\subfloat[]{\includegraphics[width=0.4\columnwidth]{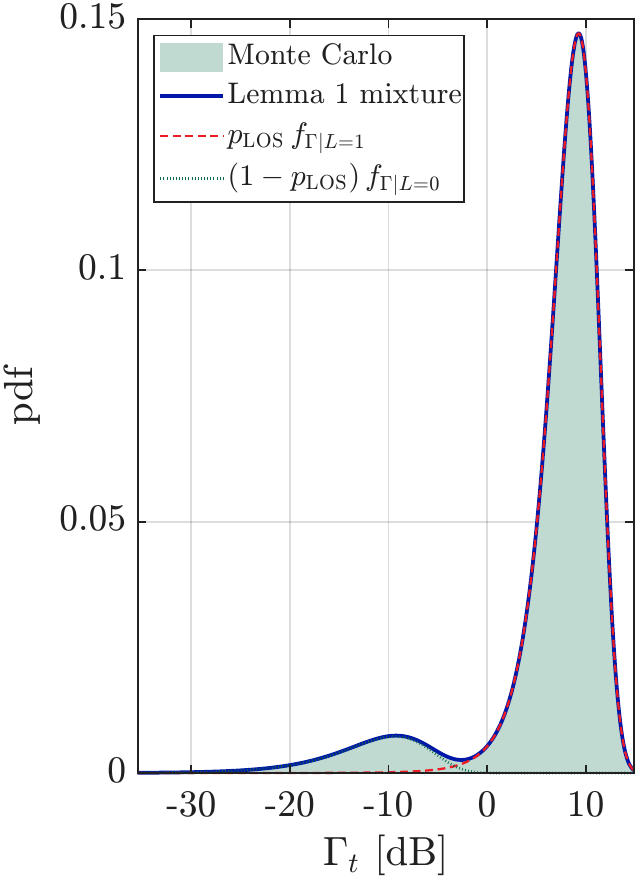}%
\label{fig:snr-validation}}
\hfil
\subfloat[]{\includegraphics[width=0.4\columnwidth]{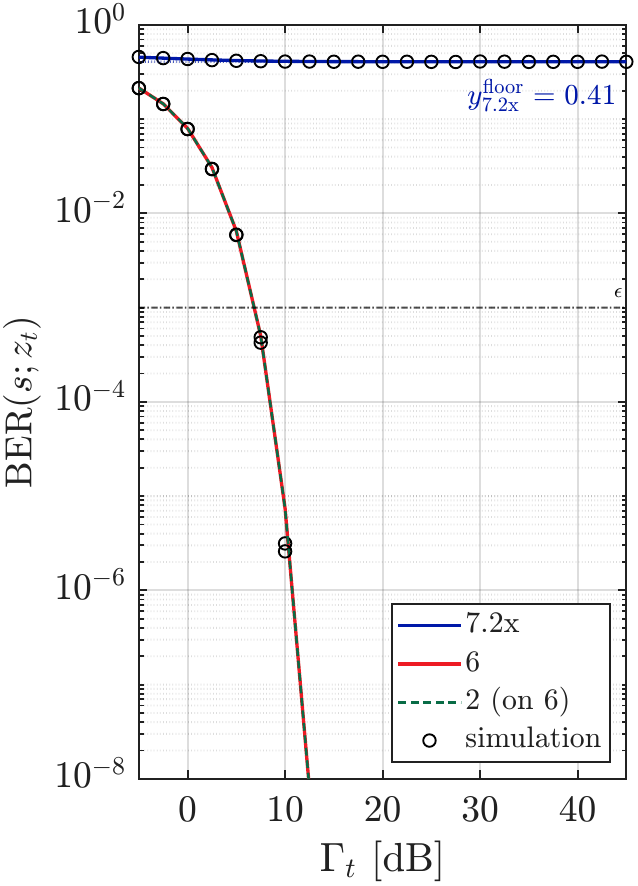}%
\label{fig:ber-floor}}
\caption{\small Validation of the analysis: (a) mixture SNR density of Lemma~\ref{lem:snr} (dB domain); (b) BER law and aging floors of Proposition~\ref{prop:ber-aging}.}
\label{fig:analysis-validation}
\end{figure}

\subsection{Delay, Doppler, and Split-Dependent CSI Age}
\label{subsec:csi-age}
The access delay is the link-plane quantity $\tau_p(t)=d_{\mathrm{SU}}(t)/c$ (6.44~ms at the $10^\circ$ edge to 2.0~ms at zenith for 600~km), and the Doppler shift is the three-dimensional projection $f_d(t)=-(f_c/c)\bigl(\vec{v}_S(t)-\vec{v}_U(t)\bigr)^\top \hat{\vec{d}}_{\mathrm{SU}}(t)$, with $\hat{\vec{d}}_{\mathrm{SU}}=\vec{d}_{\mathrm{SU}}/\|\vec{d}_{\mathrm{SU}}\|$, which remains exact for an inclined orbit, arbitrary user, and rotating Earth. The common Doppler is precompensated by the NR-NTN loops, leaving a residual $\tilde{f}_d(t)=f_d(t)-\hat{f}_d(t)$.

The placement $\varphi_s$ determines where the CSI consumed by PHY adaptation (estimation, precoding) and MAC adaptation (scheduling, MCS, HARQ) is processed, hence how old it is when applied. With $T_s^{\mathrm{age}}=t_s^{\mathrm{use}}-t^{\mathrm{obs}}$ and the feeder delay $\tau_f(t)=\|\vec{r}_S(t)-\vec{r}_G\|/c$, the two aging timescales per split are
\begin{align}
T_{s_1}^{\mathrm{PHY}} &= 2\tau_f + T_{s_1}^{\mathrm{proc,PHY}}, &
T_{s_1}^{\mathrm{MAC}} &= 2\tau_f + T_{s_1}^{\mathrm{proc,MAC}}, \nonumber\\
T_{s_2}^{\mathrm{PHY}} &= T_{s_2}^{\mathrm{proc,PHY}}, &
T_{s_2}^{\mathrm{MAC}} &= 2\tau_f + T_{s_2}^{\mathrm{proc,MAC}}, \nonumber\\
T_{s_3}^{\mathrm{PHY}} &= T_{s_3}^{\mathrm{proc,PHY}}, &
T_{s_3}^{\mathrm{MAC}} &= T_{s_3}^{\mathrm{proc,MAC}},
\label{eq:csi-ages}
\end{align}
where each $T_s^{\mathrm{proc,X}}$ collects the sub-millisecond estimation, precoding, measurement, scheduling, and queueing terms of the corresponding loop: Split 7.2x closes both loops through the feeder (twice per loop), Split 6 closes the PHY loop onboard but not the MAC loop, and Split 2 closes both onboard. Since $\tau_f$ is several milliseconds, the ordering
\begin{equation}
\label{eq:age-ordering}
T_{s_3}^{\mathrm{PHY}}=T_{s_2}^{\mathrm{PHY}} \ll T_{s_1}^{\mathrm{PHY}},\quad
T_{s_3}^{\mathrm{MAC}} \ll T_{s_2}^{\mathrm{MAC}} \approx T_{s_1}^{\mathrm{MAC}},
\end{equation}
holds throughout the visibility window, dominated by $2\tau_f(t)$.

\subsection{Aging Error, Effective SNR, and BER Law}
\label{subsec:aging-ber}
The temporal correlation of the scalar channel at residual maximum Doppler $\tilde{f}_d^{\max}$ is modeled by the Jakes form $\rho(\Delta t)\approx J_{0}(2\pi \tilde{f}_d^{\max}\Delta t)$, a conservative working model for the LOS-dominated regime. Under the first-order aging model $h_{t+T_{s}^{\mathrm{age}}}=\rho_{s}h_t+\sqrt{1-|\rho_{s}|^{2}}\,e_t$ with $\rho_s=\rho(T_s^{\mathrm{age}})$ and assuming the MMSE one-tap predictor $\hat{h}_{t+T_s^{\mathrm{age}}}=\rho_s h_t$ with ideal initial estimation, the normalized CSI error is $\varepsilon_{s}^{\mathrm{X}}=1-|\rho(T_{s}^{\mathrm{X}})|^{2}$ for $\mathrm{X}\in\{\mathrm{PHY},\mathrm{MAC}\}$, nonincreasing toward more onboard processing by \eqref{eq:age-ordering}. Normalizing the innovation to the instantaneous gain, $\mathbb{E}[|e_t|^2\mid h_t]=|h_t|^2$, as residual Doppler rotates the LOS channel itself, and treating it as self-interference, the mismatched-reception approximation yields the effective post-aging SNR $\gamma_{s}^{\mathrm{eff}}=|\rho_s|^{2}\Gamma_t\big/[1+(1-|\rho_s|^{2})\Gamma_t]$, with $\rho_s=\rho(T_{s}^{\mathrm{PHY}})$, which is strictly increasing in $\Gamma_t$ and saturates at $|\rho_s|^2/(1-|\rho_s|^2)$: stale CSI induces an error floor that no transmit-power increase can remove. The uncoded BER of split $s$ under coherent BPSK is $\mathrm{BER}(s;\vec{z}_t)=Q(\sqrt{2\gamma_{s}^{\mathrm{eff}}})$, so that, unlike formulations with an exogenous per-split SNR penalty, the split-dependence of the reliability is \emph{derived} from the placement. By \eqref{eq:age-ordering}, $\rho_{s_1}\le\rho_{s_2}=\rho_{s_3}$ and hence $\mathrm{BER}(s_1;\vec{z}_t)\ge\mathrm{BER}(s_2;\vec{z}_t)=\mathrm{BER}(s_3;\vec{z}_t)$: relocating the PHY onboard provably improves the access link, and Splits 6 and 2 are distinguished not by BER but by the MAC-level aging constraint.

\begin{proposition}
\label{prop:ber-aging}
Let $\kappa_s=1-|\rho_s|^2$ and $u_s(\gamma)=|\rho_s|^2\gamma/(1+\kappa_s\gamma)$, so that $u_s(\Gamma_t)=\gamma_s^{\mathrm{eff}}$ is the post-aging SNR of Sec.~\ref{subsec:aging-ber}. Then $u_s$ is strictly increasing with supremum $u_s^{\max}=|\rho_s|^2/\kappa_s$, so $Y_{s,t}=\mathrm{BER}(s;\vec{z}_t)$ is supported on $(y_s^{\mathrm{floor}},\tfrac12]$ with the aging-induced floor $y_s^{\mathrm{floor}}=Q\bigl(\sqrt{2|\rho_s|^{2}/(1-|\rho_s|^{2})}\bigr)$, and, for any target $\epsilon\in(0,\tfrac12)$ with $u_{\mathrm{th}}=\tfrac12(Q^{-1}(\epsilon))^2$,
\begin{equation}
\label{eq:aging-threshold}
\mathrm{BER}(s;\vec{z}_t)\le\epsilon
\iff
\Gamma_t \ \ge\ \gamma_s^{\mathrm{th}}
= u_{\mathrm{th}}\big/\bigl(|\rho_s|^2-\kappa_s u_{\mathrm{th}}\bigr)
\end{equation}
whenever $u_{\mathrm{th}}<u_s^{\max}$, the constraint being infeasible at any power otherwise. For $\rho_s=1$, \eqref{eq:aging-threshold} reduces to the pure-AWGN threshold $\gamma_s^{\mathrm{th}}=u_{\mathrm{th}}$.
\end{proposition}
\begin{proof}
Since $u_s'(\gamma)=|\rho_s|^2/(1+\kappa_s\gamma)^2>0$, $u_s$ is strictly increasing on $\gamma\ge0$ with $u_s(0)=0$ and $u_s(\gamma)\to|\rho_s|^2/\kappa_s=u_s^{\max}$ as $\gamma\to\infty$, so it maps $[0,\infty)$ bijectively onto $[0,u_s^{\max})$. As $Q$ is strictly decreasing with $Q(0)=\tfrac12$, $\mathrm{BER}(s;\vec{z}_t)=Q(\sqrt{2u_s(\Gamma_t)})$ is strictly decreasing in $\Gamma_t$ and takes values in $(Q(\sqrt{2u_s^{\max}}),\tfrac12]$, which is the stated support, the floor $y_s^{\mathrm{floor}}$ being approached but never attained as $\Gamma_t\to\infty$. For the threshold, strict monotonicity of $Q$ gives $\mathrm{BER}(s;\vec{z}_t)\le\epsilon\iff u_s(\Gamma_t)\ge u_{\mathrm{th}}$. If $u_{\mathrm{th}}\ge u_s^{\max}$ no $\Gamma_t$ satisfies this, whence infeasibility at any power. Otherwise, inverting $u=u_s(\gamma)$ yields $\gamma=u/(|\rho_s|^2-\kappa_s u)$, with positive denominator exactly when $u<u_s^{\max}$, and monotonicity gives $u_s(\Gamma_t)\ge u_{\mathrm{th}}\iff\Gamma_t\ge\gamma_s^{\mathrm{th}}$. Finally $\rho_s=1$ gives $\kappa_s=0$, hence $u_s(\gamma)=\gamma$ and $\gamma_s^{\mathrm{th}}=u_{\mathrm{th}}$.
\end{proof}
Fig.~\ref{fig:ber-floor} confirms the law against a bit-error simulation in which the receiver detects with aged CSI while the realized channel is $h_{t+T_s^{\mathrm{age}}}=\rho_s h_t+\sqrt{1-|\rho_s|^2}\,e_t$: the transparent split saturates at $y_{s_1}^{\mathrm{floor}}$, never reaching $\epsilon$ at any power, while $y_{s_2}^{\mathrm{floor}}$ and $y_{s_3}^{\mathrm{floor}}$ of the two regenerative splits coincide and fall off exponentially.

\begin{remark}[Digitized fronthaul versus analog bent-pipe]
\label{rem:analog-vs-digital}
Since the O-RU is onboard for every candidate, even $s_1$ relays a \emph{digitized} fronthaul, so the two hops decouple ($\gamma_{\mathrm{e2e}}=\Gamma_t$) rather than accumulating noise as in an analog bent pipe, $\gamma_{\mathrm{e2e}}=\gamma_f\gamma_a/(\gamma_f+\gamma_a+1)$.
\end{remark}

\section{Split Optimization and Its Exact Solution}
\label{sec:solution}
\subsection{Fronthaul Rates and Constrained Selection}
\label{subsec:fronthaul}
Split 7.2x transports frequency-domain IQ samples over the Open Fronthaul at
\begin{equation}
\label{eq:fh-split1}
R_{\mathrm{fh}}(s_1) = 2 q_{\mathrm{IQ}} N_{\mathrm{stream}} N_{\mathrm{RE}}^{\mathrm{act}} \bigl(1+\xi_{1}\bigr) \big/ \bigl(T_{\mathrm{slot}}\eta_{\mathrm{comp}}\bigr),
\end{equation}
with $q_{\mathrm{IQ}}=16$ quantization bits, $N_{\mathrm{stream}}=4$ streams, $N_{\mathrm{RE}}^{\mathrm{act}}=14 N_{sc}$ resource elements per slot ($N_{sc}=B_t/\Delta f$, $\Delta f=15$~kHz), $T_{\mathrm{slot}}=1$~ms, compression $\eta_{\mathrm{comp}}$, and overhead $\xi_1$. Splits 6 and 2 transport MAC--PHY and F1 traffic at the delivered information rate, with $R_c$ the channel code rate, $R_{\mathrm{fh}}(s_2)=R_c B(1+\xi_2)+R_{\mathrm{ctrl},2}$ and $R_{\mathrm{fh}}(s_3)=R_c B(1+\xi_3)+R_{\mathrm{ctrl},3}$, so $R_{\mathrm{fh}}(s_1)\gg R_{\mathrm{fh}}(s_2)>R_{\mathrm{fh}}(s_3)$ against a finite feeder capacity $C_{\mathrm{fh}}$.

Centralized (ground-side) processing is the architecturally preferred operating point: it preserves O-RAN pooling and upgradeability and minimizes the payload burden. With centralization order $s_1\prec s_2\prec s_3$, define the feasible set
{\small
\begin{equation}
\label{eq:feasible-set}
\mathcal{F}(\vec{z}_t) = \left\{s \in \mathcal{S} :
\begin{aligned}
&\mathrm{BER}(s;\vec{z}_t) \leq \epsilon,\\
&\varepsilon_s^{\mathrm{PHY}}(\vec{z}_t) \leq \varepsilon_{\mathrm{PHY}}^{\max},\ \
\varepsilon_s^{\mathrm{MAC}}(\vec{z}_t) \leq \varepsilon_{\mathrm{MAC}}^{\max},\\
&R_{\mathrm{fh}}(s;B_t) \leq C_{\mathrm{fh}},\ \
\text{payload limits satisfied}
\end{aligned}
\right\},
\end{equation}
}
where the onboard compute/power/memory limits are treated as satisfied constants. The selection rule is the prefer-centralized feasibility ladder: among the feasible splits, take the least regenerative one,
\begin{equation}
\label{eq:ladder-rule}
s_t^\star = \min_{\prec}\ \mathcal{F}(\vec{z}_t),
\end{equation}
escalating onboard processing only as the channel dynamics and transport limits require; if $\mathcal{F}(\vec{z}_t)=\varnothing$ the state is rejected during dataset construction. Unlike a BER-only formulation, the feasible set is not nested: aging disqualifies centralized splits under fast dynamics and capacity disqualifies 7.2x under narrow feeder margins, independently of the SNR.

Two closures make $s_t^\star$ a deterministic function of $\vec{z}_t$: the gateway lies within the regional coverage area, and the frequency-compensation loop leaves a fixed residual fraction $\zeta\in(0,1)$ ($10^{-3}$, ${\approx}45$~Hz), so that $\tau_f(t) \approx \tau_{p,t}$ and $\tilde{f}_d^{\max}(t) = \zeta|f_{d,t}|$.
All per-split ages, correlations, and aging errors are then computable from $(\tau_{p,t},f_{d,t})$ and fixed constants.

\subsection{Closed-Form Constraint Inversions and Optimal Solution}
\label{subsec:inversions}
Because $\sum_s x_s=1$ restricts any binary-program encoding of \eqref{eq:ladder-rule} to three points, the exact solution is an $O(|\mathcal{S}|)$ feasibility scan; what makes it a \emph{closed-form solution} is that every active constraint inverts analytically. (i)~\emph{Reliability:} by Proposition~\ref{prop:ber-aging}, $\mathrm{BER}(s;\vec{z}_t)\le\epsilon \iff \Gamma_t\ge\gamma_s^{\mathrm{th}}$. (ii)~\emph{CSI aging:} since $\varepsilon_s^{\mathrm{X}}\le\varepsilon_{\mathrm{X}}^{\max} \iff |\rho(T_s^{\mathrm{X}})|\ge\sqrt{1-\varepsilon_{\mathrm{X}}^{\max}}$, the Jakes model with the closures above and the principal (monotone) branch of $J_0$ gives the Doppler--delay bound
\begin{equation}
\label{eq:aging-bound}
|f_{d,t}|\; T_s^{\mathrm{X}}(\tau_{p,t})
\ \le\
J_0^{-1}\!\bigl(\sqrt{1-\varepsilon_{\mathrm{X}}^{\max}}\bigr)\big/(2\pi\zeta)
\ \triangleq \ c_{\mathrm{X}},
\end{equation}
with $T_s^{\mathrm{X}}(\tau_{p,t})=2\tau_{p,t}+T_s^{\mathrm{proc,X}}$ for the ground-closed loops; the onboard-closed loops are state-independent constants verified offline, so only $\varepsilon_{s_1}^{\mathrm{PHY}}$, $\varepsilon_{s_1}^{\mathrm{MAC}}$, $\varepsilon_{s_2}^{\mathrm{MAC}}$ are checked per sample. (iii)~\emph{Capacity:} \eqref{eq:fh-split1} is linear in $B_t$, so $R_{\mathrm{fh}}(s_1;B_t)\le C_{\mathrm{fh}}\iff B_t\le B_1^{\max}=C_{\mathrm{fh}}T_{\mathrm{slot}}\eta_{\mathrm{comp}}\Delta f\big/[28 q_{\mathrm{IQ}}N_{\mathrm{stream}}(1+\xi_1)]$, the thresholds for $s_2,s_3$ being non-binding by the rate ordering.

\begin{remark}
\label{rem:region-rule}
Under the closures of Sec.~\ref{subsec:fronthaul} and the principal-branch restriction of $J_0$, the optimal label \eqref{eq:ladder-rule} is a deterministic, piecewise-constant function of $(\Gamma_t,\tau_{p,t},f_{d,t},B_t)$:
\begin{equation}
\label{eq:region-rule}
s_t^\star=
\begin{cases}
s_1, & \vec{z}_t\in\mathcal{R}_1\cap\mathcal{A}_1\cap\mathcal{B}_1,\\
s_2, & \text{else if } \vec{z}_t\in\mathcal{R}_2\cap\mathcal{A}_2,\\
s_3, & \text{else if } \vec{z}_t\in\mathcal{R}_3,\\
\varnothing, & \text{otherwise (sample rejected)},
\end{cases}
\end{equation}
where $\mathcal{R}_s=\{\Gamma_t\ge\gamma_s^{\mathrm{th}}\}$, $\mathcal{A}_s$ are the aging regions \eqref{eq:aging-bound} active under $s$, and $\mathcal{B}_1=\{B_t\le B_1^{\max}\}$.
\end{remark}

\section{KAN-LSTM Split-Schedule Predictor}
\label{sec:kan-lstm}
\subsection{From Reactive Selection to Schedule Prediction}
The optimal solution is \emph{reactive}, but a split change is a placement reconfiguration $\varphi_{s}\rightarrow\varphi_{s'}$ whose control signaling itself traverses the feeder, so it must be initiated with a lead time $\Delta \gtrsim 2\tau_f + T_{\mathrm{rc}}$, with $T_{\mathrm{rc}}$ the second-scale non-RT RIC loop and function instantiation \cite{Polese_COMST_2023}. This motivates predicting the ground-truth label $\Delta$ steps ahead from the state history, $f_\theta:(\vec{z}_1,\dots,\vec{z}_t)\mapsto\hat{y}_{t+\Delta}\in\{1,2,3\}$, with $\Delta=0$ recovering per-sample classification. Although the instantaneous label is memoryless (Remark~\ref{rem:region-rule}), the \emph{future} label depends on the link-plane trajectory, whose elevation, delay, and Doppler trends anticipate upcoming region-boundary crossings.

\subsection{Architecture, Data, and Training}
The predictor is a compact recurrent model whose LSTM affine maps are replaced by Kolmogorov--Arnold-inspired layers \cite{liu2025kan}: we take each map to be a powered-ReLU branch with a linear residual,
\begin{equation}
\label{eq:powermlp-map}
\Psi(u) = \sigma_k\!\left(W u + b\right) + W_b\,u,
\qquad
\sigma_k(v)=\max(0,v)^k,
\end{equation}
with $k=3$. With $\tilde{\vec{z}}_t$ the standardized features (SNR in dB) and $x_t=[\tilde{\vec{z}}_t^\top, h_{t-1}^\top]^\top$, the gates are $i_t = \mathrm{sigm}(\Psi_i(x_t))$, $f_t = \mathrm{sigm}(\Psi_f(x_t))$, $o_t = \mathrm{sigm}(\Psi_o(x_t))$, and $\tilde{c}_t = \tanh(\Psi_c(x_t))$, with cell update $c_t = f_t \odot c_{t-1} + i_t \odot \tilde{c}_t$ and $h_t = o_t \odot \tanh(c_t)$, with a read-out layer of the same form \eqref{eq:powermlp-map} producing the logits $g_\theta(t)=W^{\mathrm{o}}\Psi_r(h_t)+b^{\mathrm{o}}$, softmax posterior $p_\theta(i\mid \vec{z}_{1:t})$, and $\hat{y}_{t+\Delta}=\arg\max_i [g_\theta(t)]_i$. At hidden width 16 the model retains an $\mathcal{O}(10^3)$-parameter footprint, deployable as an xApp/rApp on the gateway RIC \cite{Polese_COMST_2023,Do_RIVF_2025}, with one cell update per 1-Hz epoch.

Training data are \emph{pass sequences}: visibility windows with varying orbit orientation and user location, sampled at 1~Hz, with states $\{\vec{z}_t\}$ from the stochastic channel and labels $\{y_t\}$ from the ground-truth policy \eqref{eq:region-rule} (rejected epochs masked). Pairs $(\vec{z}_{1:t}, y_{t+\Delta})$ shift the labels by the horizon, so the target is always the exact optimal schedule; passes (not epochs) are partitioned into training/validation/test. The loss is the class-weighted cross-entropy,
\begin{equation}
\label{eq:kan-lstm-loss}
\mathcal{L}(\theta)
= -\frac{1}{\sum_p |\mathcal{T}_p|}\sum_{p}\sum_{t\in\mathcal{T}_p}
w_{y_{t+\Delta}}\,
\log p_\theta\!\left(y_{t+\Delta} \mid \vec{z}_{1:t}\right),
\end{equation}
with inverse-frequency weights $w_i\propto N/N_i$, optimized by truncated BPTT with Adam. Beyond schedule accuracy, the BER-violation probability $V_{\mathrm{BER}}=\Pr[\mathrm{BER}(\hat{s}_{t+\Delta};\vec{z}_{t+\Delta})>\epsilon]$, and the signed fronthaul-rate gap $\Delta\bar{R}_{\mathrm{fh}}=\overline{R_{\mathrm{fh}}(\hat{s}_{t+\Delta})-R_{\mathrm{fh}}(s^\star_{t+\Delta})}$, positive when a policy loads the feeder more heavily than the ground-truth policy and negative when it loads it less, we report \emph{event-level} metrics (switch-detection delay, false-switch rate), since a trivial persistence predictor scores highly on a piecewise-constant schedule by never switching.

\section{Numerical Results}
\label{sec:results}
\subsection{Scenario Parameters and Evaluation Ensembles}
All results use the 600~km S-band suburban scenario with $\epsilon=10^{-3}$, $\varepsilon_{\mathrm{PHY}}^{\max}=0.30$, $\varepsilon_{\mathrm{MAC}}^{\max}=0.50$, $C_{\mathrm{fh}}=500$~Mbps, $\zeta=10^{-3}$, and $B_t\in\{5,10,20\}$~MHz. The optimal solution is evaluated on a large Monte Carlo ensemble of access-link states, with repeated fading draws per epoch and the transmit power swept over its operating range; schedule prediction uses randomized passes sampled at 1~Hz, partitioned at the pass level into training, validation, and test sets.

\begin{figure}[t]
	\centering
	\includegraphics[width=0.95\columnwidth]{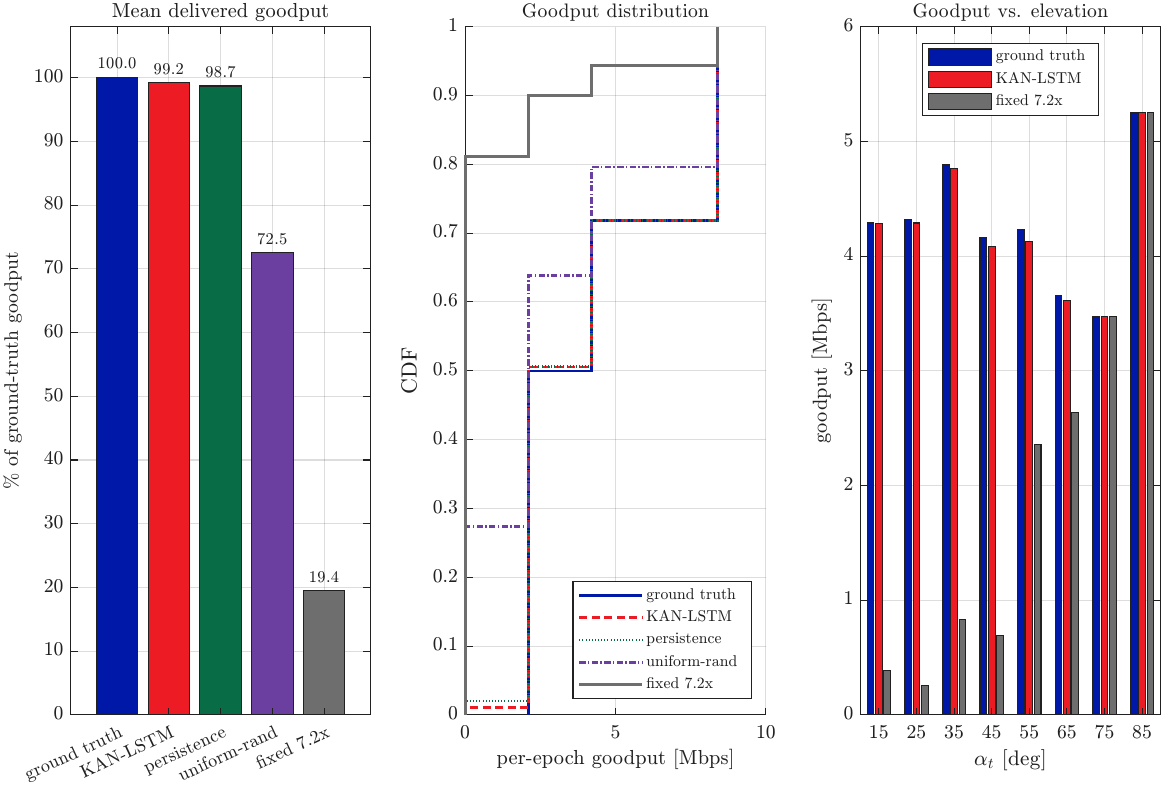}
	\caption{\small Delivered goodput over valid test epochs: mean relative to the ground-truth ceiling (left), per-epoch CDF (center), and versus elevation (right).}
	\label{fig:goodput}
\end{figure}

\begin{figure}[t]
	\centering
	\includegraphics[width=0.95\columnwidth]{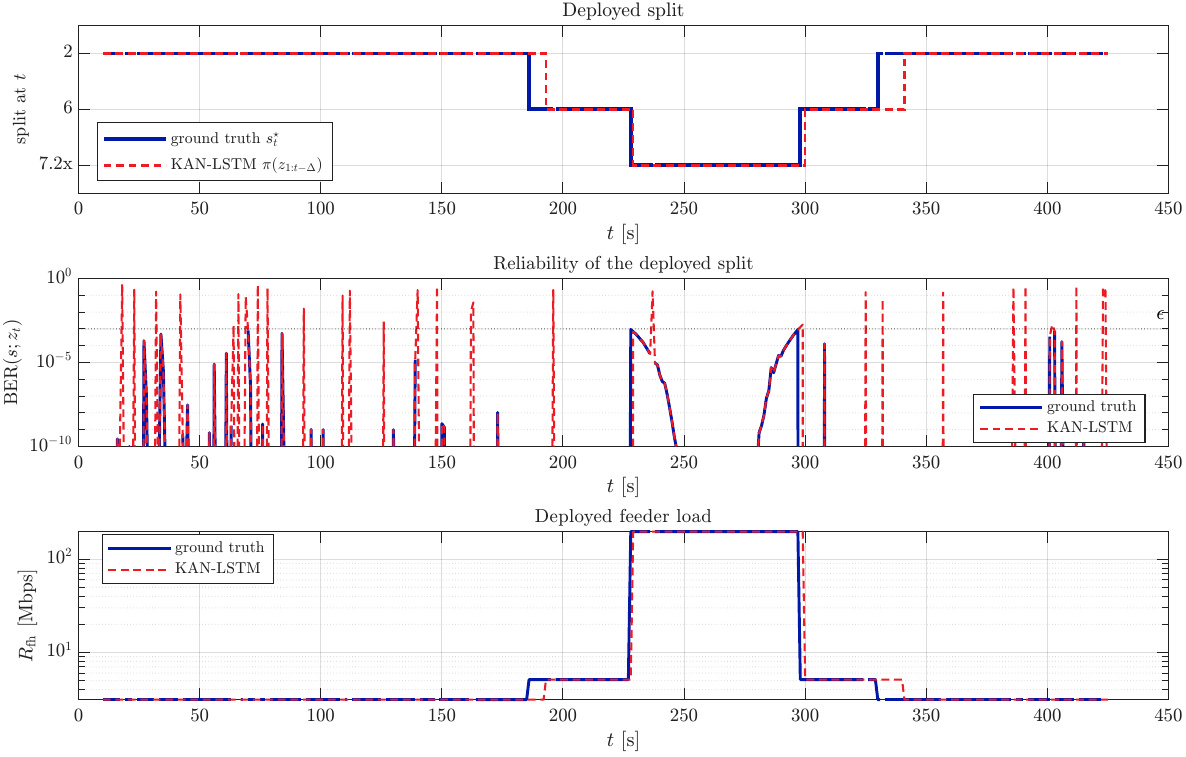}
	\caption{\small Ground-truth policy and KAN-LSTM on the same $\vec{z}_t$ sequence of a held-out pass: split, uncoded BER, and feeder load ($\Delta=10$~s ahead).}
	\label{fig:compare-time}
\end{figure}

\subsection{Exact-Solution Validation and Baseline Comparison}
Proposition~\ref{prop:ber-aging} is verified first. Over $318{,}600$ cases ($106{,}200$ states, three splits), the equivalence $\mathrm{BER}(s;\vec{z}_t)\le\epsilon\iff\Gamma_t\ge\gamma_s^{\mathrm{th}}$ of \eqref{eq:aging-threshold} holds without a single mismatch, no sample falls below its floor $y_s^{\mathrm{floor}}$, and wherever $u_{\mathrm{th}}\ge u_s^{\max}$ the target is unmet at any power. The floor separates the splits: the feeder-induced CSI age of 7.2x leaves $y_{s_1}^{\mathrm{floor}}$ far above $\epsilon$, whereas regenerating the PHY loop onboard makes the floor negligible. The closed-form rule of Remark~\ref{rem:region-rule} reproduces enumeration of \eqref{eq:ladder-rule} exactly, the label geometry being the tiling \eqref{eq:region-rule} predicts: reliability boundary vertical in $\Gamma_t$, aging boundary horizontal in the Doppler--delay product. Against the ground-truth policy, a uniform-random split frequently violates the constraints, while a random-feasible split never violates but often forfeits the centralization preference, its \emph{negative} fronthaul gap, a lighter load bought by needlessly relocating processing onboard, showing that the rate metric must be read jointly with the reliability ones. With the split-independent access rate $R_{\mathrm{acc}}(B_t)=R_c N_{\mathrm{RE}}^{\mathrm{act}}/T_{\mathrm{slot}}$, the delivered goodput $G(\pi)=\overline{R_{\mathrm{acc}}(B_t)\,\mathbf{1}[\mathrm{BER}(\hat s_t;\vec{z}_t)\le\epsilon]}$ recasts reliability as rate: the ground-truth policy's guarantee makes its goodput the outage-free ceiling, which a uniform-random split attains only partially and a fixed always-centralized 7.2x split barely at all.

\subsection{KAN-LSTM Schedule Prediction}
The KAN-LSTM (hidden width 16, $k=3$) is trained for 40 epochs with Adam at horizon $\Delta=10$~s, covering $T_{\mathrm{rc}}$ ($2\tau_f\approx13$~ms being negligible); the loss converges within a few epochs and validation accuracy plateaus without overfitting.

Table~\ref{tab:kanlstm-metrics} compares against the persistence baseline (hold the current ground-truth label for $\Delta$ seconds): the KAN-LSTM is more accurate, roughly halves the BER violations, and improves recall most on the minority centralized classes (11/26/63\% of valid test epochs). Fig.~\ref{fig:goodput} places all policies on a single delivered-rate axis: both predictors recover nearly all of the ground-truth policy's outage-free goodput ceiling, the KAN-LSTM slightly more, whereas a uniform-random split loses a substantial share and a fixed always-centralized 7.2x split collapses, most severely at low elevation where the feeder delay, and hence the CSI age, is largest. At the event level, about half of the switching events are predicted within the matching window; the missed ones are fading-driven micro-switches, and the schedule error concentrates within a few seconds of a ground-truth switch and nearly vanishes away from it, i.e., timing jitter at region-boundary crossings. Fig.~\ref{fig:compare-time} shows a held-out pass: the predicted schedule anticipates the staircase $2\rightarrow6\rightarrow7.2\mathrm{x}\rightarrow6\rightarrow2$ with second-scale offsets.

\begin{table}[t]
\caption{\small KAN-LSTM versus the persistence baseline.}
\label{tab:kanlstm-metrics}
\centering
\scriptsize
\setlength{\tabcolsep}{2.5pt}
\begin{tabular}{@{}lccccc@{}}
\toprule
Policy & Acc.\ [\%] & $V_{\mathrm{BER}}$ [\%] & $\Delta\bar{R}_{\mathrm{fh}}$ [Mbps] & \makecell{Recall\\7.2x/6/2 [\%]} & \makecell{Good.\\{}[\%]} \\
\midrule
KAN-LSTM     & 96.81 & 0.99 & $+2.76$ & 99.1 / 94.4 / 97.4 & 99.2 \\
Persistence  & 90.75 & 2.02 & $-0.12$ & 84.6 / 83.7 / 95.6 & 98.7 \\
\bottomrule
\end{tabular}
\end{table}

Fig.~\ref{fig:wireless-perf} translates the policies into deployed wireless performance: the uncoded-BER distributions coincide in the body, the ground-truth policy's CCDF truncating exactly at $\epsilon$ by construction, the KAN-LSTM showing only a thin tail and persistence a heavier one; the post-aging SNR distributions and feeder loads nearly coincide, so the learned policy differs from the ground-truth policy only in the reliability tail. Single-epoch inference takes microseconds, negligible at 1~Hz.

\begin{figure}[t]
\centering
\includegraphics[width=0.95\columnwidth]{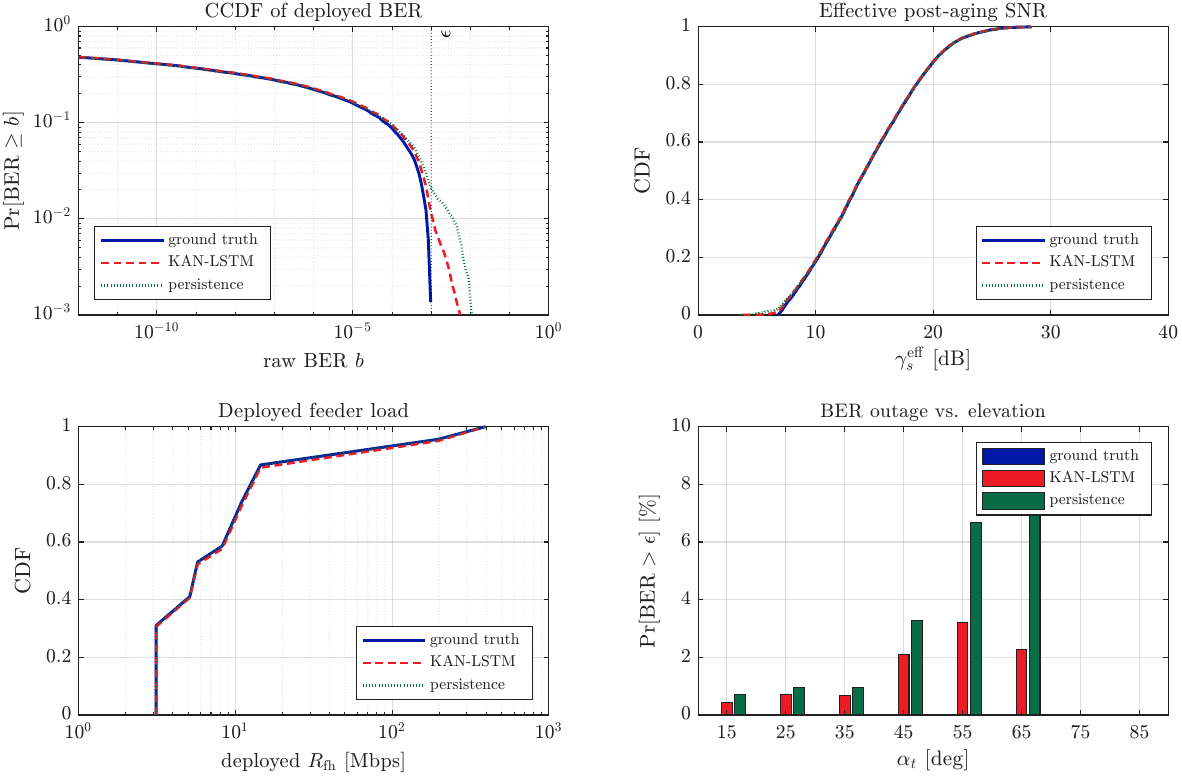}
\caption{\small Deployed wireless performance: CCDF of uncoded BER, CDFs of post-aging SNR and feeder load, and BER outage versus elevation.}
\label{fig:wireless-perf}
\end{figure}

\section{Conclusion}
\label{sec:conclusion}
In this paper, we proposed a layered framework for functional-split selection in an O-RAN NTN, in which each split places the RAN at a distinct point of the transparent-to-regenerative axis. An exact link-plane geometry and a split-dependent CSI-aging mechanism make reliability a provable consequence of the placement rather than an exogenous penalty, and the resulting prefer-centralized feasibility problem admits a closed-form region rule as its exact solution, in full agreement with exhaustive enumeration. Because this rule is reactive, a compact KAN-LSTM anticipates its schedule seconds ahead, recovering nearly all of the ground-truth goodput and outperforming naive and reactive baselines.


\bibliographystyle{IEEEtran}
\bibliography{references}

@IEEEtranBSTCTL{IEEEtran:BSTcontrol,
  CTLuse_forced_etal       = "yes",
  CTLmax_names_forced_etal = "5",
  CTLnames_show_etal       = "1"
}

@ARTICLE{Jamshed_MCOM_2026,
	author={Jamshed, Muhammad Ali and Saad, Malik Muhammad and Mohsin, Muhammad Ahmed and Kim, Dongkyun and Dobre, Octavia A. and Yanikomeroglu, Halim and Mohjazi, Lina},
	journal={{IEEE} Communications Magazine},
	title={Net-Zero {6G} from {E}arth to Orbit: Sustainable Design of Integrated Terrestrial and Non-Terrestrial Networks},
	year={2026},
	volume={},
	number={},
	pages={1-10},
	doi={10.1109/MCOM.001.2500625},
	ISSN={1558-1896},
	month={},}

@article{Polese_COMST_2023,
  author  = {Polese, Michele and Bonati, Leonardo and D'Oro, Salvatore and Basagni, Stefano and Melodia, Tommaso},
  title   = {Understanding {O-RAN}: Architecture, Interfaces, Algorithms, Security, and Research Challenges},
  journal = {{IEEE} Communications Surveys \& Tutorials},
  volume  = {25},
  number  = {2},
  pages   = {1376--1411},
  year    = {2023},
  eprint  = {2202.01032},
  eprinttype = {arXiv},
  doi     = {10.1109/COMST.2023.3239220},
  arxiv   = {https://arxiv.org/abs/2202.01032}
}

@ARTICLE{Alam_COMST_2026,
  author={Alam, Khurshid and Habibi, Mohammad Asif and Tammen, Matthias and Krummacker, Dennis and Saad, Walid and Renzo, Marco Di and Melodia, Tommaso and Costa-Pérez, Xavier and Debbah, Mérouane and Dutta, Ashutosh and Schotten, Hans D.},
  journal={{IEEE} Communications Surveys \& Tutorials},
  title={A Comprehensive Tutorial and Survey of {O-RAN}: Exploring Slicing-Aware Architecture, Deployment Options, Use Cases, and Challenges},
  year={2026},
  volume={28},
  number={},
  pages={1637-1678},
  doi={10.1109/COMST.2025.3598406},
  ISSN={1553-877X},
  month={},}

@ARTICLE{Lin2026_arXiv,
	title         = "The emerging {3GPP} {6G} {RAN} architecture: An overview",
	author        = "Lin, Xingqin",
	journal       = "arXiv preprint",
	month         =  jul,
	year          =  2026,
	archivePrefix = "arXiv",
	primaryClass  = "cs.NI",
	eprint        = "2607.10543",
  doi     = {10.48550/arXiv.2607.10543},
  url     = {https://doi.org/10.48550/arXiv.2607.10543},
}

@INPROCEEDINGS{Laouadi_ICC_2025,
  author    = {Laouadi, Abir and Langar, Rami},
  title     = {Energy and {QoS}-Aware Functional Split Selection in {5G} and Beyond {O-RAN} Networks},
  booktitle = {{IEEE} International Conference on Communications ({ICC})},
  pages     = {3594--3599},
  year      = {2025},
  doi       = {10.1109/ICC52391.2025.11161774},
}

@article{Murti_TWC_2022,
  author  = {Murti, Fahri Wisnu and Ali, Samad and Latva-aho, Matti},
  title   = {Constrained Deep Reinforcement Based Functional Split Optimization in Virtualized {RAN}s},
  journal = {{IEEE} Transactions on Wireless Communications},
  volume  = {21},
  number  = {11},
  pages   = {9850--9864},
  year    = {2022},
  doi     = {10.1109/TWC.2022.3179811},
}

@ARTICLE{Carbonara_OJCOMS_2025,
  author={Carbonara, Salvatore and Olivieri, Marco and Rago, Arcangela and Sciancalepore, Vincenzo and Piro, Giuseppe and Boggia, Gennaro and Grieco, Luigi Alfredo},
  journal={{IEEE} Open Journal of the Communications Society}, 
  title={Hands-On Solutions for Testing Integrated Terrestrial and Non-Terrestrial Networks: A Comprehensive Survey}, 
  year={2025},
  volume={6},
  number={},
  pages={10729-10760},
  doi={10.1109/OJCOMS.2025.3646364},
  ISSN={2644-125X},
  month={},}

@INPROCEEDINGS{Do_RIVF_2025,
	author={Do, Tri Nhu},
	booktitle={2025 {RIVF} International Conference on Computing and Communication Technologies ({RIVF})}, 
	title={{AI}-{Open}-{RAN} for Non-Terrestrial Networks}, 
	year={2025},
	volume={},
	number={},
	pages={908-913},
	doi={10.1109/RIVF68649.2025.11365028},
	ISSN={2473-0130},
	month={Dec},}

@inproceedings{
liu2025kan,
title={{KAN}: Kolmogorov{\textendash}{A}rnold Networks},
author={Ziming Liu and Yixuan Wang and Sachin Vaidya and Fabian Ruehle and James Halverson and Marin Soljacic and Thomas Y. Hou and Max Tegmark},
booktitle={The Thirteenth International Conference on Learning Representations},
year={2025},
doi={10.48550/arXiv.2404.19756},
}

@INPROCEEDINGS{Mahyoub_IWCMC_2025,
  author={Mahyoub, Mohammed and Muhaidat, Sami and Yanikomeroglu, Halim and Kurt, Gunes Karabulut},
  booktitle={2025 International Wireless Communications and Mobile Computing ({IWCMC})},
  title={Risk-Aware Slicing-Based Security Functions Allocation in {LEO} Satellite Networks},
  year={2025},
  volume={},
  number={},
  pages={728-733},
  doi={10.1109/IWCMC65282.2025.11059478},
  ISSN={2376-6506},
  month={May},}

\end{document}